\documentclass[11pt]{article}
\usepackage{hyperref}
\usepackage{times}  
\usepackage{mathpazo}
\usepackage{amssymb,amsmath,amsthm}
\usepackage{epsfig}
\usepackage{amsmath}
\usepackage{fullpage}

\newtheorem{theorem}{Theorem}[section]

\newtheorem{definition}[theorem]{Definition}

\newtheorem{lemma}[theorem]{Lemma}

\newtheorem{corollary}[theorem]{Corollary}

\newtheorem{remark}[theorem]{Remark}

\newcommand{\qedsymb}{\hfill{\rule{2mm}{2mm}}}
\renewenvironment{proof}[1][]{\begin{trivlist}
\item[\hspace{\labelsep}{\bf\noindent Proof#1:\/}] }{\qedsymb\end{trivlist}}

\def\calC{{\cal C}}
\def\calG{{\cal G}}

\def\calC{{\cal C}}

\def\R{\mathbb{R}}

\def\N{\mathbb{N}}

\newcommand{\DRamsey}{\mathsf{DRamsey}}
\newcommand{\Ramsey}{\mathsf{Ramsey}}

\newcommand{\MAIS}{\mathrm{MAIS}}
\newcommand{\MES}{\mathrm{MES}}
\newcommand{\CC}{\mathrm{CC}}

\newcommand{\NP}{\mathsf{NP}}

\newcommand{\eps}{\epsilon}
\renewcommand{\epsilon}{\varepsilon}

\newcommand{\Fset}{\mathbb{F}}         

\begin{document}

\title{{\bf Improved Approximation Algorithms for Index Coding}}
\author{
Dror Chawin\thanks{School of Computer Science, The Academic College of Tel Aviv-Yaffo, Tel Aviv 61083, Israel. Research supported by the Israel Science Foundation (grant No.~1218/20).}
\and
Ishay Haviv\footnotemark[1]
}

\date{}

\maketitle

\begin{abstract}
The index coding problem is concerned with broadcasting encoded information to a collection of receivers in a way that enables each receiver to discover its required data based on its side information, which comprises the data required by some of the others.
Given the side information map, represented by a graph in the symmetric case and by a digraph otherwise, the goal is to devise a coding scheme of minimum broadcast length.

We present a general method for developing efficient algorithms for approximating the index coding rate for prescribed families of instances.
As applications, we obtain polynomial-time algorithms that approximate the index coding rate of graphs and digraphs on $n$ vertices to within factors of $O(n/\log^2 n)$ and $O(n/\log n)$ respectively. This improves on the approximation factors of $O(n/\log n)$ for graphs and $O(n \cdot \log \log n/\log n)$ for digraphs achieved by Blasiak, Kleinberg, and Lubetzky (IEEE Trans. Inform. Theory,~2013).
For the family of quasi-line graphs, we exhibit a polynomial-time algorithm that approximates the index coding rate to within a factor of $2$. This improves on the approximation factor of $O(n^{2/3})$ achieved by Arbabjolfaei and Kim (ISIT,~2016) for graphs on $n$ vertices taken from certain sub-families of quasi-line graphs.

Our approach is applicable for approximating a variety of additional graph and digraph quantities to within the same approximation factors.
Specifically, it captures every graph quantity sandwiched between the independence number and the clique cover number and every digraph quantity sandwiched between the maximum size of an acyclic induced sub-digraph and the directed clique cover number.
\end{abstract}

\section{Introduction}

The index coding problem, introduced in 1998 by Birk and Kol~\cite{BirkKol98}, is a central problem in network information theory, motivated by various applications ranging from satellite and wireless communications to network coding and distributed storage (see, e.g.,~\cite{ArbabjolfaeiK18}). The problem involves a sender that holds an $n$-symbol message $x \in \Sigma^n$ over some alphabet $\Sigma$ and $n$ receivers $R_1, \ldots, R_n$, where the $i$th receiver $R_i$ is interested in the $i$th symbol $x_i$. Every receiver $R_i$ is given as side information the restriction $x_{N(i)}$ of $x$ to the indices that lie in some set $N(i) \subseteq [n] \setminus \{i\}$. The side information map is represented by a digraph $G$ on the vertex set $[n]$ that includes a directed edge $(i,j)$ if $j \in N(i)$, that is, if the receiver $R_i$ knows $x_j$. The sender wishes to broadcast to the receivers a short encoded message that enables each receiver $R_i$ to discover $x_i$ based on the transmitted message and on the given side information.

An index code for a digraph $G$ on the vertex set $[n]$ over an alphabet $\Sigma$ is formally defined as an encoding function $E:\Sigma^n \rightarrow \Sigma_E$ for some alphabet $\Sigma_E$, such that for each $i \in [n]$, there exists a decoding function $g_i:\Sigma_E \times \Sigma^{|N(i)|} \rightarrow \Sigma$ satisfying $g_i(E(x),x_{N(i)}) = x_i$ for all $x \in \Sigma^n$. Note that the encoding length in bits of such an index code is $\lceil \log_2 |\Sigma_E| \rceil$.
For a side information digraph $G$ and for an integer $t \geq 1$, let $\beta_t(G)$ denote the smallest possible encoding length in bits of an index code for $G$ over an alphabet $\Sigma$ of size $|\Sigma| = 2^t$.
The index coding rate of $G$, denoted by $\beta(G)$, expresses the average asymptotic number of encoding bits needed per bit in each input block for the side information digraph $G$.
It is therefore defined by $\beta(G) = \lim_{t \rightarrow \infty}{\frac{\beta_t(G)}{t}}$, where the existence of the limit follows by combining the sub-additivity of $\beta_t$ with Fekete's lemma (see, e.g.,~\cite[Lemma 11.6]{BookFekete}).
The quantities $\beta_t$ and $\beta$ are naturally extended to (undirected) graphs by replacing every undirected edge by two oppositely directed edges.

A simple upper bound on the index coding rate of a digraph is its (directed) clique cover number. Here, a clique of a digraph $G=(V,E)$ is a set $C \subseteq V$ such that every two distinct vertices of $C$ are connected by two oppositely directed edges. For a digraph $G$ and a clique cover $\calC$ of $G$, consider the encoding function over $\Sigma=\{0,1\}$ that for every clique $C \in \calC$ includes the xor of the bits associated with $C$. This encoding function forms an index code for $G$, because each receiver is able to discover its bit given its side information and the xor associated with a clique that covers its vertex. We refer to such an index code as a clique cover index code. It follows that every digraph $G$ satisfies $\beta(G) \leq \beta_1(G) \leq \CC(G)$, where $\CC(G)$ stands for the minimum size of a clique cover of $G$. This upper bound on the index coding rate has multiple stronger variants, e.g., a local relaxation and a fractional relaxation based on linear programming (see, e.g.,~\cite[Section~6]{ArbabjolfaeiK18}).
On the other hand, the index coding rate $\beta(G)$ of a digraph $G$ is known to satisfy $\beta(G) \geq \MAIS(G)$, where $\MAIS(G)$ stands for the maximum size of an acyclic induced sub-digraph of $G$ (see~\cite{BBJK06,BlasiakKL13}). Note that in the undirected setting, the above bounds can be written as $\alpha(G) \leq \beta(G) \leq \overline{\chi}(G)$, where $\alpha(G)$ and $\overline{\chi}(G)$ denote, respectively, the independence and clique cover numbers of a graph $G$.

From a computational point of view, the difficulty of determining the value of $\beta(G)$ for a given graph or digraph $G$ has motivated researchers to suggest efficient approximation algorithms for the problem. A result of Blasiak, Kleinberg, and Lubetzky~\cite{BlasiakKL13} implies that there exists a polynomial-time algorithm that given a digraph $G$ on $n$ vertices approximates the value of $\beta(G)$ to within a factor of $O(n \cdot \frac{\log \log n}{\log n})$, improving on the trivial factor of $n$. It was further noted in~\cite[Remark~V.2]{BlasiakKL13} that for graphs on $n$ vertices, it is possible to achieve in polynomial time a better approximation factor of $O(\frac{n}{\log n})$.
The proofs in~\cite{BlasiakKL13} further imply that the fractional clique cover number of a digraph on $n$ vertices approximates the index coding rate to within a factor of $O(n \cdot \frac{\log \log n}{\log n})$, and that the clique cover number of a graph on $n$ vertices approximates the index coding rate to within a factor of $O(\frac{n}{\log n})$ (see also~\cite[Propositions~9.1 and~9.3]{ArbabjolfaeiK18}).
The challenge of improving these approximation factors in the directed and undirected settings was explicitly raised in~\cite[Open Problem~9.1]{ArbabjolfaeiK18}.
By combining an approach of~\cite{BlasiakKL13} with results from Ramsey theory~\cite{BelmonteHHRS14}, Arbabjolfaei and Kim~\cite{ArbabjolfaeiK16} showed that for certain families of graphs, the clique cover number approximates the index coding rate to within tighter factors. In particular, they obtained a factor of $O(n^{2/3})$ for graphs on $n$ vertices from the families of line graphs and fuzzy circular interval graphs (see, e.g.,~\cite{ClawFree05}).

\subsection{Our Contribution}

In this paper, we offer a general method for developing efficient algorithms for approximating the index coding rate for prescribed families of instances.
To do so, we consider the algorithmic problem, which might be of independent interest, of finding economical clique covers of digraphs, where the approximation guarantee is measured with respect to the maximum size of an acyclic induced sub-digraph of the digraph at hand.
This is in contrast to the standard setting, where the size of the produced clique cover is measured with respect to the digraph's clique cover number.
Our results on this problem yield efficient approximation algorithms for a variety of graph and digraph quantities.
In what follows, we present those results and describe their consequences for the index coding problem.

\subsubsection{Clique cover vs. Acyclic induced sub-digraph}

A family of graphs or digraphs is called hereditary if it is closed under removal of vertices (but not necessarily under removal of edges).
Consider the problem that given a digraph $G$ on $n$ vertices, taken from a prescribed hereditary digraph family $\calG$, asks to find a clique cover of $G$ whose size does not exceed $\gamma \cdot \MAIS(G)$ for as small as possible $\gamma = \gamma(n)$.
In the present paper, we exhibit a general method for tackling this problem, from both the existential and algorithmic perspectives, where the guaranteed factor $\gamma$ depends on the Ramsey numbers of the family $\calG$ (see Definition~\ref{def:Ramsey}).
Note that under this framework, one can also consider families of (undirected) graphs, represented as digraphs in which adjacent vertices are connected by two oppositely directed edges. In this representation, the maximum size of an acyclic induced sub-digraph is simply the independence number of the corresponding graph.

Before describing our general method, let us mention a well-understood special case of the above problem, where $\calG$ is the family of all graphs.
Erd{\H{o}}s~\cite{Erdos67} proved in 1967 that for every graph $G$ on $n$ vertices, it holds that $\overline{\chi}(G) \leq O(\frac{n}{\log^2 n}) \cdot \alpha(G)$.
The bound is optimal up to a multiplicative constant because, as is well known, a random graph $G$ on $n$ vertices satisfies $\alpha(G) = \Theta(\log n)$ and $\overline{\chi}(G) = \Theta(\frac{n}{\log n})$ (see, e.g.,~\cite[Chapter~10]{AlonS16}).
In the early nineties, Boppana and Halld{\'{o}}rsson~\cite{BoppanaH92} proved the following algorithmic form of Erd{\H{o}}s's result.

\begin{theorem}[{\cite[Theorem~3]{BoppanaH92}}]\label{thm:BH92}
There exists a polynomial-time algorithm that given a graph $G$ on $n$ vertices finds a clique cover of $G$ of size $O(\frac{n}{\log^2 n}) \cdot \alpha(G)$.
\end{theorem}

To state our general method, we need the following definition of Ramsey numbers, which extends a definition of Belmonte, Heggernes, van~'t Hof, Rafiey, and Saei~\cite{BelmonteHHRS14} to the directed setting.
\begin{definition}[Ramsey Numbers]\label{def:Ramsey}
For a graph family $\calG$ and for two integers $s,t \geq 1$, let $R_\calG(s,t)$ denote the smallest integer $r$ such that every member of $\calG$ on at least $r$ vertices has a clique of size $s$ or an independent set of size $t$.
Similarly, for a digraph family $\calG$ and for two integers $s,t \geq 1$, let $\vec{R}_\calG(s,t)$ denote the smallest integer $r$ such that every member of $\calG$ on at least $r$ vertices has a clique of size $s$ or an acyclic induced sub-digraph of size $t$.
\end{definition}
\noindent
Throughout the paper, for a digraph $G$ and for a set of vertices $I$, we let $G[I]$ denote the sub-digraph of $G$ induced by $I$. The same notation is used for graphs.

Our method is described by the following theorem, which generalizes the approach of~\cite{BoppanaH92}.

\begin{theorem}\label{thm:Intro_CC_general}
Let $\calG$ be a hereditary family of digraphs, and let $Q: \N \times \N \rightarrow \N$ be a function such that $\vec{R}_\calG(s,t) \leq Q(s,t)$ for all integers $s,t \geq 1$.
Let $f_Q: \N \rightarrow \N$ denote the function defined by
\[ f_Q(n) = \min_{s,t \in \N} ~\{ s \cdot t \mid Q(s+1,t+1) > n \}.\]
Then for every digraph $G \in \calG$ on $n$ vertices, it holds that
\begin{eqnarray}\label{eq:Introbound_CC_Ramsey}
\CC(G) \leq \Big (  \sum_{i=1}^{n}{\frac{1}{f_Q(i)}} \Big ) \cdot \MAIS(G).
\end{eqnarray}
Suppose further that there exists a polynomial-time algorithm that given a digraph $G \in \calG$ on $n$ vertices finds a clique $C$ of $G$ and a set $I$ for which $G[I]$ is acyclic, such that for all integers $s,t \geq 1$, if $n \geq Q(s,t)$ then either $|C| \geq s$ or $|I| \geq t$. Then there exists a polynomial-time algorithm that given a digraph $G \in \calG$ on $n$ vertices finds a clique cover of $G$ whose size does not exceed the bound in~\eqref{eq:Introbound_CC_Ramsey}.
\end{theorem}
\noindent
Note that if one is interested only in the existential statement of the theorem, given in~\eqref{eq:Introbound_CC_Ramsey}, then the function $Q$ can be taken to be equal to $\vec{R}_\calG$. However, to obtain the algorithmic statement of the theorem, it might be needed to consider some function $Q$ with larger values.

We demonstrate the usefulness of Theorem~\ref{thm:Intro_CC_general} by applying it to several families of graphs and digraphs.
Let us first mention that the aforementioned results of~\cite{Erdos67} and~\cite{BoppanaH92} (Theorem~\ref{thm:BH92}) can be derived by applying Theorem~\ref{thm:Intro_CC_general} to the family of all graphs, using a classical argument from Ramsey theory due to Erd{\H{o}}s and Szekeres~\cite{ErdosS35}.
For an analogue result for digraphs, we combine Theorem~\ref{thm:Intro_CC_general} with a directed variant of the argument of~\cite{ErdosS35} and prove the following.

\begin{theorem}\label{thm:IntroDirectedGen}
There exists a polynomial-time algorithm that given a digraph $G$ on $n$ vertices finds a clique cover of $G$ of size $O(\frac{n}{\log n}) \cdot \MAIS(G)$.
\end{theorem}
\noindent
The guarantee of the algorithm given by Theorem~\ref{thm:IntroDirectedGen} is optimal up to a multiplicative constant, because there exist digraphs on $n$ vertices whose ratio between the clique cover number and the maximum size of an acyclic induced sub-digraph is $\Theta(\frac{n}{\log n})$. This indeed follows from a result of Erd{\H{o}}s and Moser~\cite{ErdosMoser64}, who proved that there exist tournaments $G$ on $n$ vertices, i.e., digraphs with exactly one edge connecting every pair of vertices, satisfying $\MAIS(G) \leq 2 \cdot \lfloor \log_2 n \rfloor +1$. Since a tournament has no clique of size larger than $1$, those digraphs $G$ satisfy $\CC(G)=n$, and thus attain the claimed ratio.

We further apply Theorem~\ref{thm:Intro_CC_general} to hereditary digraph families $\calG$ whose Ramsey numbers $\vec{R}_\calG(s,t)$ grow polynomially with $s$ and $t$ (see Theorem~\ref{thm:CC_Ramsey}).
As an application, we improve the upper bounds of~\cite{ArbabjolfaeiK16} on the ratio between the clique cover and independence numbers for hereditary graph families with polynomially bounded Ramsey numbers (see Corollary~\ref{cor:CC_Ramsey_U} and the discussion that follows it).
In particular, for hereditary graph families $\calG$ satisfying $R_\calG(s,t) = O(s \cdot t)$, we show that every graph $G \in \calG$ on $n$ vertices satisfies $\overline{\chi}(G) \leq O(\log n) \cdot \alpha(G)$.
This particular setting is motivated by a paper of Belmonte et al.~\cite{BelmonteHHRS14}, who proved that $R_\calG(s,t) = O(s \cdot t)$ for $\calG$ being either the family of line graphs or the family of fuzzy circular interval graphs (see, e.g.,~\cite{ClawFree05}).

In fact, for these two aforementioned graph families we improve on the logarithmic multiplicative term deduced from Theorem~\ref{thm:Intro_CC_general} and obtain a much stronger bound along with an algorithmic result.
To do so, we consider the family of quasi-line graphs, those graphs where the neighborhood of every vertex can be partitioned into two cliques.
This family of graphs, which plays a central role in the structural characterization of claw-free graphs due to Chudnovsky and Seymour~\cite{ClawFree05}, contains the family of line graphs as well as the family of fuzzy circular interval graphs.
We prove the following theorem.

\begin{theorem}\label{thm:IntroQuasi}
There exists a polynomial-time algorithm that given a quasi-line graph $G$ finds a clique cover of $G$ of size at most $2 \cdot \alpha(G)$.
\end{theorem}
\noindent
The guarantee of the algorithm given by Theorem~\ref{thm:IntroQuasi} is essentially optimal, even for line graphs, because the ratio between the clique cover and independence numbers of such graphs can be arbitrarily close to $2$ when the number of vertices grows (see Remark~\ref{remark:2tight}).
Since every graph $G$ satisfies $\alpha(G) \leq \overline{\chi}(G)$, Theorem~\ref{thm:IntroQuasi} implies that it is possible to approximate the clique cover number of quasi-line graphs to within a factor of $2$ in polynomial time. This extends a recent result on line graphs due to Daneshpajouh, Meunier, and Mizrahi~\cite{DaneshpajouhMM21}.

\subsubsection{Approximation algorithms for Index Coding}

The previous section offers efficient algorithms that given a digraph $G$ on $n$ vertices, taken from some specified digraph family, find a clique cover of $G$ of size at most $\gamma \cdot \MAIS(G)$ for some $\gamma = \gamma(n)$.
Consequently, these algorithms can be used to approximate to within a factor of $\gamma$ any digraph quantity $\psi$ that satisfies $\MAIS(G) \leq \psi(G) \leq \CC(G)$ for every digraph $G$.
In the undirected setting, those algorithms can be used to approximate any graph quantity $\psi$ that satisfies $\alpha(G) \leq \psi(G) \leq \overline{\chi}(G)$ for every graph $G$.
It therefore follows that the algorithms from the previous section can be used to efficiently approximate the index coding rate $\beta$ of a given graph or digraph. Moreover, the same approximation guarantees are achievable for various additional index coding measures, e.g., the scalar capacity $\beta_1$, the $\beta^\star$ capacity due to Alon et al.~\cite{AlonLSWH08}, the minrank parameter over a field $\Fset$ due to Haemers~\cite{Haemers78} (which characterizes the optimal length of linear index codes over $\Fset$~\cite{BBJK06}), and many more. The approximation also applies to graph and digraph quantities with the above sandwich properties that are not directly related to index coding, such as the Shannon capacity of graphs~\cite{Shannon56} and the Sperner capacity of (complement) digraphs~\cite{KornerS92}.

For concreteness, we state below our results on approximating the index coding rate $\beta$.
We start with the case of general graphs.
As a consequence of the result of Boppana and Halld{\'{o}}rsson~\cite{BoppanaH92}, stated earlier as Theorem~\ref{thm:BH92}, we obtain the following.

\begin{theorem}\label{thm:IntroICundir}
There exists a polynomial-time algorithm that given a graph $G$ on $n$ vertices finds a clique cover index code for $G$ of length $O(\frac{n}{\log^2 n}) \cdot \beta(G)$, and in particular, approximates the value of $\beta(G)$ to within a factor of $O(\frac{n}{\log^2 n})$.
\end{theorem}
\noindent
Note that Theorem~\ref{thm:IntroICundir} improves on the approximation factor reported in~\cite[Remark~V.2]{BlasiakKL13} and in~\cite[Propositions~9.1]{ArbabjolfaeiK18} by a multiplicative term of $\log n$.

We proceed with the case of general digraphs. As a consequence of Theorem~\ref{thm:IntroDirectedGen}, we obtain the following.

\begin{theorem}\label{thm:IntroICdir}
There exists a polynomial-time algorithm that given a digraph $G$ on $n$ vertices finds a clique cover index code for $G$ of length $O(\frac{n}{\log n}) \cdot \beta(G)$, and in particular, approximates the value of $\beta(G)$ to within a factor of $O(\frac{n}{\log n})$.
\end{theorem}
\noindent
The algorithm given by Theorem~\ref{thm:IntroICdir} improves on the algorithm of Blasiak et al.~\cite{BlasiakKL13} for digraphs in two respects.
Firstly, the approximation factor is improved by a multiplicative term of $\log \log n$.
Secondly, our algorithm produces an index code that corresponds to a clique cover of the input digraph rather than to a fractional clique cover.
Therefore, the algorithm applies to additional index coding measures such as the scalar capacity $\beta_1$ and the minrank over a field $\Fset$, which are not bounded from above by the fractional clique cover.
On the other hand, the algorithm of~\cite{BlasiakKL13} enjoys the advantage that it extends to a generalized version of the index coding problem, which allows the number of receivers $m$ to be larger than the number of symbols $n$ (see Section~\ref{sec:generalIC}). We extend Theorem~\ref{thm:IntroICdir} to this generalized setting, however, this extension does not preserve the approximation factor. Yet, it produces index codes whose rate is smaller than the rate guaranteed by the algorithm of~\cite{BlasiakKL13} whenever the number of receivers $m$ satisfies $m = o(n \cdot \log \log n)$. For the precise statement, see Theorem~\ref{thm:genIC}.

Finally, we consider the restriction of the index coding problem to the family of quasi-line graphs.
As a consequence of Theorem~\ref{thm:IntroQuasi}, we obtain the following.

\begin{theorem}\label{thm:IntroICquasi}
There exists a polynomial-time algorithm that given a quasi-line graph $G$ finds a clique cover index code for $G$ of length at most $2 \cdot \beta(G)$, and in particular, approximates the value of $\beta(G)$ to within a factor of $2$.
\end{theorem}
\noindent
As mentioned earlier, the family of quasi-line graphs contains the family of line graphs and the family of fuzzy circular interval graphs.
Therefore, Theorem~\ref{thm:IntroICquasi} improves, in the existential and algorithmic manners, the approximation factor of $O(n^{2/3})$ achieved in~\cite{ArbabjolfaeiK16} for graphs on $n$ vertices from these families.

It would be intriguing to determine whether the approximation factors guaranteed by Theorems~\ref{thm:IntroICundir},~\ref{thm:IntroICdir}, and~\ref{thm:IntroICquasi} can be improved. A viable strategy for addressing this challenge would be to employ stronger lower and upper bounds on the index coding rate than the ones utilized here (namely, the maximum size of an acyclic induced sub-digraph and the clique cover number).

\subsection{Related Work}

We gather here several algorithmic and hardness results from the literature regarding graph quantities considered in the present paper.
The independence and clique cover numbers of graphs on $n$ vertices can be approximated to within a factor of $O(n \cdot \frac{(\log \log n)^2}{\log^3 n})$ in polynomial time~\cite{Feige04,Halldorsson93}, and it is $\NP$-hard to approximate them to within a factor of $n^{1-\eps}$ for any $\eps>0$~\cite{Zuckerman07}.

For the index coding problem, it was shown in~\cite{LangbergS08} that assuming a variant of the Unique Games Conjecture, it is hard to approximate the value of $\beta_t(G)$ for an input graph $G$ to within any constant factor for any fixed integer $t \geq 1$. It was further shown in~\cite{ChawinH23} that it is $\NP$-hard to approximate the optimal length of a linear index code over any fixed finite field to within any constant factor (see also~\cite{ChawinH22}).

It is worth noting that it is $\NP$-hard to approximate the minimum size of a vertex cover of triangle-free graphs to within a factor of $1.36$~\cite{AwasthiCKS15}.
As observed in~\cite{DaneshpajouhMM21}, hardness for the vertex cover problem on triangle-free graphs implies hardness for the clique cover problem on line graphs.
To see this, recall that the line graph of a graph $G$ is the graph whose vertices are the edges of $G$, with an edge between every two edges of $G$ that share a common vertex. It can be verified that if $G$ is triangle-free, then the minimum size of a vertex cover of $G$ is equal to the minimum size of a clique cover of its line graph. The result of~\cite{AwasthiCKS15} therefore implies that it is $\NP$-hard to approximate the clique cover number of line graphs (and thus of quasi-line graphs) to within a factor of $1.36$.

\subsection{Outline}
The rest of the paper is organized as follows.
In Section~\ref{sec:digraphs}, we present our general method for finding clique covers of digraphs and prove Theorem~\ref{thm:Intro_CC_general}.
We then apply the theorem to several families of digraphs, and in particular, prove Theorem~\ref{thm:IntroICdir}.
We also present the method for graphs and show that it generalizes Theorem~\ref{thm:BH92} of~\cite{BoppanaH92} (with the full details given in Appendix~\ref{app:undirected}).
In Section~\ref{sec:quasi}, we present our algorithm for finding clique covers of quasi-line graphs and prove Theorem~\ref{thm:IntroQuasi}.
Finally, in Section~\ref{sec:generalIC}, we use Theorem~\ref{thm:IntroICdir} to offer an algorithm for a generalized version of the index coding problem.

\section{Clique Covering of Digraphs}\label{sec:digraphs}

In this section, we present our general method for finding clique covers of digraphs, confirm Theorem~\ref{thm:Intro_CC_general}, and demonstrate its applicability.
We start with the following lemma, which produces clique covers of digraphs from a given hereditary family and analyzes their size through a standard argument (see, e.g.,~\cite[Lemma~2]{Halldorsson93}).

\begin{lemma}\label{lemma:IterativeCC}
Let $\calG$ be a hereditary family of digraphs, and let $f: \N \rightarrow \R^+$ be a monotone non-decreasing function.
Suppose that every digraph $G \in \calG$ on $n$ vertices has a clique of size at least $\frac{f(n)}{\MAIS(G)}$.
Then for every digraph $G \in \calG$ on $n$ vertices, it holds that
\begin{eqnarray}\label{eq:bound_CC}
\CC(G) \leq  \Big ( \sum_{i=1}^{n}{\frac{1}{f(i)}} \Big ) \cdot \MAIS(G).
\end{eqnarray}
Moreover, if there exists a polynomial-time algorithm that given a digraph $G \in \calG$ on $n$ vertices finds a clique of $G$ of size at least $\frac{f(n)}{\MAIS(G)}$, then there exists a polynomial-time algorithm that given a digraph $G \in \calG$ on $n$ vertices finds a clique cover of $G$ whose size does not exceed the bound in~\eqref{eq:bound_CC}.
\end{lemma}

\begin{proof}
Let $\calG$ be a hereditary family of digraphs, and let $f: \N \rightarrow \R^+$ be a function as in the lemma.
Consider some digraph $G \in \calG$ on $n$ vertices.
We define a clique cover of $G$ by repeatedly choosing a clique of $G$ and removing its vertices.

More specifically, we define a sequence of induced sub-digraphs $G_1, \ldots, G_r$ of $G$ and a sequence of cliques $C_1, \ldots, C_r$ of $G$ as follows.
Set $G_1 = G$.
For each $j \geq 1$, if $G_j$ has at least one vertex, let $n_j$ denote the number of its vertices, and use the fact that $G_j \in \calG$ to obtain that there exists a clique $C_j$ of $G_j$ of size
\begin{eqnarray}\label{eq:lemmaCj}
|C_j| \geq \frac{f(n_j)}{\MAIS(G_j)} \geq \frac{f(n_j)}{\MAIS(G)},
\end{eqnarray}
where the second inequality holds because $G_j$ is an induced sub-digraph of $G$.
We proceed to the next iteration of the process with the digraph $G_{j+1}$ obtained from $G_{j}$ by removing the vertices of $C_{j}$.
Since $\calG$ is hereditary, it holds that $G_{j+1} \in \calG$.
Note that $G_{j+1}$ has $n_j-|C_j|$ vertices.

Letting $r$ denote the number of cliques produced until no vertices remain, the sets $C_1, \ldots, C_r$ form a clique cover of $G$ of size $r$, hence $\CC(G) \leq r$.
Combining~\eqref{eq:lemmaCj} with the assumption that the function $f$ is monotone non-decreasing, we obtain that for each $j \in [r]$, it holds that
\[ 1 = |C_j| \cdot \frac{1}{|C_j|} \leq |C_j| \cdot \frac{1}{f(n_j)} \cdot \MAIS(G) \leq \bigg (\sum_{i=0}^{|C_j|-1}{\frac{1}{f(n_j-i)}} \bigg ) \cdot \MAIS(G).\]
By summing the above over all $j \in [r]$, we obtain that
\[ \CC(G) \leq r  \leq \Bigg ( \sum_{j=1}^{r} { \sum_{i=0}^{|C_j|-1}{\frac{1}{f(n_j-i)}} }\Bigg ) \cdot \MAIS(G) = \Big ( \sum_{i=1}^{n}{\frac{1}{f(i)}} \Big ) \cdot \MAIS(G),\]
where the equality holds because $n_1 = n$, $n_{j+1} = n_j-|C_j|$ for all $j \in [r-2]$, and $n_r = |C_r|$.

Finally, observe that the iterative process described above, whose number of iterations does not exceed the number of vertices of $G$, confirms the algorithmic statement of the lemma and completes the proof.
\end{proof}

Equipped with Lemma~\ref{lemma:IterativeCC}, we are ready to prove Theorem~\ref{thm:Intro_CC_general} (recall Definition~\ref{def:Ramsey}).

\begin{proof}[ of Theorem~\ref{thm:Intro_CC_general}]
Let $\calG$ be a hereditary family of digraphs, and let $Q: \N \times \N \rightarrow \N$ be a function such that $\vec{R}_\calG(s,t) \leq Q(s,t)$ for all integers $s,t \geq 1$.
Recall that $f_Q: \N \rightarrow \N$ is the function defined by
\[ f_Q(n) = \min_{s,t \in \N} ~\{ s \cdot t \mid Q(s+1,t+1) > n \},\]
and notice that it is monotone non-decreasing.

We first claim that every digraph $G \in \calG$ on $n$ vertices has a clique of size at least $\frac{f_Q(n)}{\MAIS(G)}$.
To see this, denote by $\omega(G)$ the maximum size of a clique of $G$, and let us observe that
\begin{eqnarray}\label{eq:Q(,)}
Q(\omega(G)+1, \MAIS(G)+1) > n.
\end{eqnarray}
Indeed, otherwise it would follow that
\[\vec{R}_\calG(\omega(G)+1,\MAIS(G)+1) \leq Q(\omega(G)+1,\MAIS(G)+1) \leq n,\]
which implies that $G$ has a clique of size larger than $\omega(G)$ or an acyclic induced sub-digraph of size larger than $\MAIS(G)$, in contradiction.
Now, by the definition of the function $f_Q$, it follows from~\eqref{eq:Q(,)} that $f_Q(n) \leq \omega(G) \cdot \MAIS(G)$, hence there exists a clique of $G$ of size at least $\frac{f_Q(n)}{\MAIS(G)}$.
This allows us to apply Lemma~\ref{lemma:IterativeCC} and to obtain that for every digraph $G \in \calG$ on $n$ vertices, it holds that
\begin{eqnarray}\label{eq:bound_CC_proof}
\CC(G) \leq \Big (  \sum_{i=1}^{n}{\frac{1}{f_Q(i)}} \Big ) \cdot \MAIS(G).
\end{eqnarray}

For the algorithmic statement of the theorem, suppose that there exists a polynomial-time algorithm that given a digraph $G \in \calG$ on $n$ vertices finds a clique $C$ of $G$ and a set $I$ for which $G[I]$ is acyclic, such that for all integers $s,t \geq 1$, if $n \geq Q(s,t)$ then either $|C| \geq s$ or $|I| \geq t$.
Fix a digraph $G \in \calG$ on $n$ vertices.
While running on $G$, the sets $C$ and $I$ returned by the given algorithm must satisfy $Q(|C|+1,|I|+1) > n$, hence $f_Q(n) \leq |C| \cdot |I|$.
It thus follows that the clique $C$ returned by the algorithm satisfies $|C| \geq \frac{f_Q(n)}{|I|} \geq \frac{f_Q(n)}{\MAIS(G)}$.
Therefore, using the algorithmic statement of Lemma~\ref{lemma:IterativeCC}, it follows that there exists a polynomial-time algorithm that given a digraph $G \in \calG$ on $n$ vertices finds a clique cover of $G$ whose size does not exceed the bound in~\eqref{eq:bound_CC_proof}.
This completes the proof.
\end{proof}

\subsection{General digraphs}\label{sec:applications}

Our first application of Theorem~\ref{thm:Intro_CC_general} concerns the family of all digraphs.
Namely, we prove Theorem~\ref{thm:IntroDirectedGen}, which says that there exists a polynomial-time algorithm that given a digraph $G$ on $n$ vertices finds a clique cover of $G$ of size $O(\frac{n}{\log n}) \cdot \MAIS(G)$.
To do so, we need the following lemma, whose proof relies on a directed variant of an argument from Ramsey theory due to Erd{\H{o}}s and Szekeres~\cite{ErdosS35}.

\begin{lemma}\label{lemma:RamseyD}
Let $Q: \N \times \N \rightarrow \N$ be the function defined by $Q(s,t) = \binom{s+t-2}{s-1} \cdot 2^{t-1}$.
There exists a polynomial-time algorithm that given a digraph $G$ on $n$ vertices finds a clique $C$ of $G$ and a set $I$ for which $G[I]$ is acyclic, such that for all integers $s,t \geq 1$, if $n \geq Q(s,t)$ then either $|C| \geq s$ or $|I| \geq t$.
\end{lemma}

\begin{proof}
We define a recursive algorithm, called $\DRamsey$, that given a digraph $G=(V,E)$ returns a pair $(C,I)$ of subsets of $V$. The algorithm is defined as follows.

\begin{enumerate}
  \item\label{itm:1} If $V = \emptyset$, then return the pair $(\emptyset,\emptyset)$.
  \item\label{itm:2} Choose an arbitrary vertex $u \in V$.
  \item\label{itm:3} $(C_1, I_1) \leftarrow \DRamsey(G[V_1])$, where $V_1 = \{ v \in V \mid (u,v) \in E~\mbox{and}~(v,u) \in E\}$.
  \item\label{itm:4} $(C_2, I_2) \leftarrow \DRamsey(G[V_2])$, where $V_2 = \{ v \in V \mid (u,v) \notin E~\mbox{and}~(v,u) \in E\}$.
  \item\label{itm:5} $(C_3, I_3) \leftarrow \DRamsey(G[V_3])$, where $V_3 = V \setminus (V_1 \cup V_2 \cup \{u\})$.
  \item\label{itm:6} Let $C$ be a largest set among $C_1 \cup \{u\}$, $C_2$, and $C_3$, let $I$ be a largest set among $I_1$, $I_2 \cup \{u\}$, and $I_3 \cup \{u\}$, and return the pair $(C,I)$.
\end{enumerate}

Observe that for any choice of a vertex $u$ in Step~\ref{itm:2} of the $\DRamsey$ algorithm, the sets $V_1$, $V_2$, and $V_3$ defined in Steps~\ref{itm:3},~\ref{itm:4}, and~\ref{itm:5} do not include $u$ and are pairwise disjoint.
This implies that throughout the run of the algorithm on a digraph $G$, every recursive call chooses in Step~\ref{itm:2} a distinct vertex of $G$, hence the total number of recursive calls does not exceed the number of vertices in $G$.
Since the number of operations made by the algorithm for preparing the three inputs of the recursive calls and for combining the returned pairs to produce the output of the algorithm is polynomial, we conclude that the algorithm can be implemented in polynomial time.

We first prove that on every input digraph $G=(V,E)$, the $\DRamsey$ algorithm returns a pair $(C,I)$ such that $C$ is a clique of $G$ and $G[I]$ is acyclic.
This is shown by induction on the number of vertices of $G$. If $G$ has no vertices, this trivially holds by Step~\ref{itm:1} of the algorithm.
Otherwise, the algorithm chooses in Step~\ref{itm:2} an arbitrary vertex $u \in V$, and for each $i \in [3]$, it calls the $\DRamsey$ algorithm recursively on the sub-digraph $G[V_i]$ to obtain a pair $(C_i,I_i)$, where $V_1,V_2,V_3$ are defined in Steps~\ref{itm:3},~\ref{itm:4}, and~\ref{itm:5} of the algorithm. By the inductive assumption, for each $i \in [3]$, $C_i$ is a clique of $G[V_i]$, and $I_i$ induces an acyclic sub-digraph of $G[V_i]$.

By the definition of $V_1$, the vertex $u$ is connected to all vertices of $V_1$ with edges in both directions, hence $C_1 \cup \{u\}$ is a clique of $G$. It thus follows that the set $C$, defined in Step~\ref{itm:6} as a largest set among $C_1 \cup \{u\}$, $C_2$, and $C_3$, is a clique of $G$, as required.
By the definition of $V_2$, the vertex $u$ has no directed edges to the vertices of $V_2$.
Since $I_2$ induces an acyclic sub-digraph of $G$, it follows that $I_2 \cup \{u\}$ induces such a sub-digraph as well.
Similarly, by the definition of $V_3$, no vertex of $V_3$ has a directed edge to the vertex $u$.
Since $I_3$ induces an acyclic sub-digraph of $G$, it follows that $I_3 \cup \{u\}$ induces such a sub-digraph as well.
It thus follows that the set $I$, defined in Step~\ref{itm:6} as a largest set among $I_1$, $I_2 \cup \{u\}$, and $I_3 \cup \{u\}$, induces an acyclic sub-digraph of $G$, as required.

We finally analyze the sizes of the sets returned by the $\DRamsey$ algorithm.
To do so, for every pair of integers $s,t \geq 1$, let $\vec{P}(s,t)$ denote the smallest integer $p$ such that on every digraph on at least $p$ vertices, the $\DRamsey$ algorithm is guaranteed to return a clique of size at least $s$ or a set of vertices that induces an acyclic sub-digraph of size at least $t$.
We claim that for all $s,t \geq 2$, it holds that
\begin{eqnarray}\label{eq:DRamsey}
\vec{P}(s,t) \leq \vec{P}(s-1,t)+ 2 \cdot \vec{P}(s,t-1)-2.
\end{eqnarray}
To see this, consider a digraph $G$ on $p \geq \vec{P}(s-1,t)+ 2 \cdot \vec{P}(s,t-1)-2$ vertices.
While running on $G$, the $\DRamsey$ algorithm calls itself recursively on three sub-digraphs of $G$, induced by three sets $V_1$, $V_2$, and $V_3$, which satisfy $|V_1|+|V_2|+|V_3|=p-1$. Our assumption on $p$ implies that either $|V_1| \geq \vec{P}(s-1,t)$, or $|V_2| \geq \vec{P}(s,t-1)$, or $|V_3| \geq \vec{P}(s,t-1)$.
By the definition of $\vec{P}(s,t)$, the three pairs $(C_1,I_1)$, $(C_2,I_2)$, and $(C_3,I_3)$ returned by the recursive calls satisfy $|C_1| \geq s-1$ or $|I_1| \geq t$ in the first case, $|C_2| \geq s$ or $|I_2| \geq t-1$ in the second case, and $|C_3| \geq s$ or $|I_3| \geq t-1$ in the third.
By the definition of the algorithm, it follows that in each of these cases, the $\DRamsey$ algorithm returns on $G$ a clique of size at least $s$ or a set of vertices that induces an acyclic sub-digraph of size at least $t$, as required.

It remains to verify that $\vec{P}(s,t) \leq Q(s,t)$ for all $s,t \geq 1$.
This can be checked by induction on $s+t$.
First, if either $s=1$ or $t=1$, then $\vec{P}(s,t)=1 \leq Q(s,t)$.
Next, for the induction step, combine~\eqref{eq:DRamsey} with the inductive assumption on $s+t-1$ to obtain that
\begin{eqnarray*}
\vec{P}(s,t) &\leq& \vec{P}(s-1,t)+ 2 \cdot \vec{P}(s,t-1)-2 \\
&\leq& Q(s-1,t)+ 2 \cdot Q(s,t-1) \\
&=& \binom{s+t-3}{s-2} \cdot 2^{t-1} + 2 \cdot \binom{s+t-3}{s-1} \cdot 2^{t-2} \\
&=& \Bigg ( \binom{s+t-3}{s-2} + \binom{s+t-3}{s-1} \Bigg ) \cdot 2^{t-1} \\
&=& \binom{s+t-2}{s-1} \cdot 2^{t-1} = Q(s,t),
\end{eqnarray*}
where the third equality holds by Pascal's identity.
This completes the proof.
\end{proof}

Now, we combine Theorem~\ref{thm:Intro_CC_general} with Lemma~\ref{lemma:RamseyD} to prove Theorem~\ref{thm:IntroDirectedGen}.

\begin{proof}[ of Theorem~\ref{thm:IntroDirectedGen}]
Let $Q: \N \times \N \rightarrow \N$ be the function given in Lemma~\ref{lemma:RamseyD}, defined by $Q(s,t) = \binom{s+t-2}{s-1} \cdot 2^{t-1}$.
Notice that for all integers $s,t \geq 1$, it holds that
\[Q(s+1,t+1) = \binom{s+t}{s} \cdot 2^t \leq 2^{s+t} \cdot 2^t = 2^{s+2t}.\]
It follows, for every integer $n$, that if $Q(s+1,t+1) > n$ then $2^{s+2t} > n$, which implies that either $s > \lfloor (\log_2 n)/3 \rfloor$ or $t > \lfloor (\log_2 n)/3 \rfloor$. This yields that the function $f_Q$ associated with $Q$ in Theorem~\ref{thm:Intro_CC_general} satisfies
\begin{eqnarray}\label{eq:f_Q_D}
f_Q(n) > \lfloor (\log_2 n)/3 \rfloor.
\end{eqnarray}
Equipped with the algorithm given by Lemma~\ref{lemma:RamseyD}, we apply Theorem~\ref{thm:Intro_CC_general} to the family of all digraphs.
It follows that there exists a polynomial-time algorithm that given a digraph $G$ on $n$ vertices finds a clique cover of $G$ whose size does not exceed
\[ \Big (\sum_{i=1}^{n}{\frac{1}{f_Q(i)}} \Big ) \cdot \MAIS(G) \leq  O \Big (\frac{n}{\log n} \Big ) \cdot \MAIS(G),\]
where the above inequality can be verified using~\eqref{eq:f_Q_D}. This completes the proof.
\end{proof}

\subsection{Digraph families with polynomial Ramsey numbers}

Our next application of Theorem~\ref{thm:Intro_CC_general} concerns digraph families whose Ramsey numbers grow polynomially.

\begin{theorem}\label{thm:CC_Ramsey}
Let $\calG$ be a hereditary family of digraphs such that for some constants $c$ and $a \geq 1$, it holds that $\vec{R}_\calG(s,t) \leq c \cdot (s \cdot t)^a$ for all integers $s,t \geq 1$. Then for every digraph $G \in \calG$ on $n$ vertices, it holds that $\CC(G) \leq h(n) \cdot \MAIS(G)$ for a function $h:\N \rightarrow \N$ satisfying $h(n) = \Theta(\log n)$ if $a=1$, and $h(n) = \Theta(n^{1-1/a})$ otherwise.
\end{theorem}

\begin{proof}
Let $\calG$ be a hereditary family of digraphs as in the theorem.
Let $Q: \N \times \N \rightarrow \N$ be the function defined by $Q(s,t) = \vec{R}_\calG(s,t)$ for all $s,t$.
Notice that for all integers $s,t \geq 1$, it holds that
\[Q(s+1,t+1) \leq c \cdot (s+1)^a \cdot (t+1)^a \leq c \cdot (2s)^a \cdot (2t)^a.\]
It therefore follows, for every integer $n$, that if $Q(s+1,t+1)>n$ then $s \cdot t > n^{1/a} / (4 \cdot c^{1/a})$.
This yields that the function $f_Q$ associated with $Q$ in Theorem~\ref{thm:Intro_CC_general} satisfies $f_Q(n) \geq \Omega(n^{1/a})$.
Applying Theorem~\ref{thm:Intro_CC_general} to the family $\calG$, it follows that for every digraph $G \in \calG$ on $n$ vertices, it holds that $\CC(G) \leq h(n) \cdot \MAIS(G)$ for a function $h:\N \rightarrow \N$ satisfying $h(n) = \Theta(\sum_{i=1}^{n}{\frac{1}{i^{1/a}}})$.
A standard integral calculation implies that for $a=1$, it holds that $h(n) = \Theta(\log n)$, and for $a>1$, it holds that $h(n) = \Theta(n^{1-1/a})$. This completes the proof.
\end{proof}

\subsection{Graph families}\label{sec:graphs}

Our method for finding clique covers of digraphs can also be used in the undirected setting.
To see this, for any given hereditary family $\calG$ of graphs, apply Theorem~\ref{thm:Intro_CC_general} to the family of digraphs obtained from the graphs of $\calG$ by replacing every edge by two oppositely directed edges. Observe that this operation preserves the cliques of each graph of $\calG$ and, in addition, it transforms the notion of independent sets of the graph to the notion of acyclic induced sub-digraphs of the corresponding digraph.
This allows us to derive the following undirected variant of Theorem~\ref{thm:Intro_CC_general}.

\begin{theorem}\label{thm:CC_general_Undir}
Let $\calG$ be a hereditary family of graphs, and let $Q: \N \times \N \rightarrow \N$ be a function such that $R_\calG(s,t) \leq Q(s,t)$ for all integers $s,t \geq 1$.
Let $f_Q: \N \rightarrow \N$ denote the function defined by
\[ f_Q(n) = \min_{s,t \in \N} ~\{ s \cdot t \mid Q(s+1,t+1) > n \}.\]
Then for every graph $G \in \calG$ on $n$ vertices, it holds that
\begin{eqnarray}\label{eq:bound_CC_Ramsey_Undir}
\overline{\chi}(G) \leq \Big (  \sum_{i=1}^{n}{\frac{1}{f_Q(i)}} \Big ) \cdot \alpha(G).
\end{eqnarray}
Suppose further that there exists a polynomial-time algorithm that given a graph $G \in \calG$ on $n$ vertices finds a clique $C$ of $G$ and an independent set $I$ of $G$, such that for all integers $s,t \geq 1$, if $n \geq Q(s,t)$ then either $|C| \geq s$ or $|I| \geq t$. Then there exists a polynomial-time algorithm that given a graph $G \in \calG$ on $n$ vertices finds a clique cover of $G$ whose size does not exceed the bound in~\eqref{eq:bound_CC_Ramsey_Undir}.
\end{theorem}

Theorem~\ref{thm:CC_general_Undir} can be used to derive the algorithmic result of Boppana and Halld{\'{o}}rsson~\cite{BoppanaH92}, stated earlier as Theorem~\ref{thm:BH92}. For completeness, we provide the details in Appendix~\ref{app:undirected}.
We further state here the following corollary of Theorem~\ref{thm:CC_Ramsey} for the undirected setting.

\begin{corollary}\label{cor:CC_Ramsey_U}
Let $\calG$ be a hereditary family of graphs such that for some constants $c$ and $a \geq 1$, it holds that $R_\calG(s,t) \leq c \cdot (s \cdot t)^a$ for all integers $s,t \geq 1$. Then for every graph $G \in \calG$ on $n$ vertices, it holds that~ $\overline{\chi}(G) \leq h(n) \cdot \alpha(G)$ for a function $h:\N \rightarrow \N$ satisfying $h(n) = \Theta(\log n)$ if $a=1$, and $h(n) = \Theta(n^{1-1/a})$ otherwise.
\end{corollary}
\noindent
It was shown in~\cite[Theorem~1]{ArbabjolfaeiK16} that for graphs on $n$ vertices taken from a hereditary family $\calG$ whose Ramsey numbers satisfy $R_\calG(s,t) \leq O(s^a \cdot t^b)$ for constants $a,b \geq 1$, the ratio between the clique cover and independence numbers is bounded by $O(n^{1-1/(a+b+1)})$.
Our Corollary~\ref{cor:CC_Ramsey_U} improves this bound to $O(n^{1-1/\max(a,b)})$ when $\max(a,b)>1$, and to $O(\log n)$ when $a=b=1$.

\section{Clique Covering of Quasi-Line Graphs}\label{sec:quasi}

A graph $G$ is called a quasi-line graph if the neighborhood of every vertex of $G$ can be partitioned into two cliques.
Note that all line graphs are quasi-line graphs, and that all quasi-line graphs are claw-free, however, the converse of each of these statements is false (see~\cite{ClawFree05}).
We present here the simple proof of Theorem~\ref{thm:IntroQuasi}, which asserts that there exists a polynomial-time algorithm that given a quasi-line graph $G$ finds a clique cover of $G$ of size at most $2 \cdot \alpha(G)$. As mentioned earlier, the theorem strengthens a result on line graphs given in~\cite[Proposition~5]{DaneshpajouhMM21}.

\begin{proof}[ of Theorem~\ref{thm:IntroQuasi}]
Consider the algorithm that given a quasi-line graph $G$ acts as follows.
The algorithm maintains a collection $\calC$ of cliques of $G$ initiated as $\calC = \emptyset$.
As long as $G$ has at least one vertex, the algorithm chooses an arbitrary vertex $u$ in $G$, adds to $\calC$ at most two cliques of $G$ that cover the closed neighborhood of $u$, and removes their vertices from $G$. Finally, when $G$ has no vertices, the algorithm returns the obtained collection $\calC$.

We first show that the algorithm is guaranteed to find a clique cover of the input graph in polynomial time.
Let $G$ be a quasi-line graph. By definition, the subgraph induced by the neighborhood of any vertex $u$ can be partitioned into at most two cliques of $G$, which can be found in polynomial time (say, by $2$-coloring the complement of the subgraph induced by the neighborhood of $u$).
By adding $u$ to one of those cliques, we obtain the cliques that cover the closed neighborhood of $u$, as needed in every iteration of the algorithm.
Since the family of quasi-line graphs is hereditary, we are allowed to repeat this process after removing the covered vertices from $G$.
Once no vertices remain in the graph, the collection $\calC$ forms a clique cover of $G$. Since the number of vertices of $G$ decreases in every iteration, the number of iterations does not exceed the number of vertices of the input graph, hence the total running time is polynomial.

We next analyze the size of the clique cover returned by the algorithm.
For a quasi-line graph $G$, let $r$ denote the number of iterations applied by the above algorithm while running on $G$, and let $u_1, \ldots, u_r$ denote the vertices chosen throughout these iterations. Since every iteration adds to the collection $\calC$ at most two cliques, it follows that the returned $\calC$ satisfies $|\calC| \leq 2 \cdot r$. We further observe that the vertices $u_1, \ldots, u_r$ form an independent set in $G$. This indeed holds because whenever a vertex is chosen by the algorithm, its closed neighborhood is removed from the graph. Therefore, the clique cover $\calC$ returned by the algorithm satisfies $|\calC| \leq 2 \cdot r \leq 2 \cdot \alpha(G)$, as required.
\end{proof}

\begin{remark}\label{remark:2tight}
The guarantee of the algorithm given by Theorem~\ref{thm:IntroQuasi} is essentially optimal, even for line graphs.
To see this, let $G$ denote the line graph of the complete graph on $n$ vertices for an integer $n \geq 4$.
Consequently, the vertices of $G$ are all the $2$-subsets of $[n]$, where two such subsets are adjacent if they have a non-trivial intersection.
It follows that $G$ is the complement of the Kneser graph $K(n,2)$, hence it holds that $\alpha(G) = \lfloor n/2 \rfloor$ and $\overline{\chi}(G) = n-2$ (see~\cite{LovaszKneser}).
\end{remark}

\section{Generalized Index Coding}\label{sec:generalIC}

In this section, we consider a generalized version of the index coding problem, studied in~\cite{BlasiakKL13} and defined as follows.
An instance $H$ of the problem consists of integers $m$ and $n$ such that $m \geq n$, and for each $j \in [m]$, an element $r(j) \in [n]$ and a set $N(j) \subseteq [n] \setminus \{r(j)\}$.
Here, a sender holds an $n$-symbol message $x \in \Sigma^n$ over some alphabet $\Sigma$, and $m$ denotes the number of receivers. The $j$th receiver is interested in the symbol $x_{r(j)}$ of $x$ and has the restriction $x_{N(j)}$ of $x$ to the indices of $N(j)$ as side information. We naturally assume that for each $i \in [n]$ there exists some $j \in [m]$ with $r(j)=i$. An index code for $H$ over $\Sigma$ is an encoding function $E: \Sigma^n \rightarrow \Sigma_E$ for some alphabet $\Sigma_E$, such that for each $j \in [m]$, there exists a decoding function $g_j: \Sigma_E \times \Sigma^{|N(j)|} \rightarrow \Sigma$ satisfying $g_j(E(x),x_{N(j)})=x_{r(j)}$ for all $x \in \Sigma^n$.
The encoding length in bits of this index code is $\lceil \log_2 |\Sigma_E| \rceil$.
For an integer $t \geq 1$, we let $\beta_t(H)$ denote the smallest possible encoding length in bits of an index code for $H$ over an alphabet $\Sigma$ of size $|\Sigma|=2^t$. The index coding rate of $H$, denoted by $\beta(H)$, is defined by $\beta(H) = \lim_{t \rightarrow \infty}{\frac{\beta_t(H)}{t}}$.
Note that the special case of the problem with $m=n$ coincides with the standard index coding problem on digraphs.
Note further that $\beta(H) \leq \beta_1(H) \leq n$, where the second inequality follows by the trivial index code that transmits the entire message.

Let $\calC$ be a collection of subsets of $[n]$, such that for each $j \in [m]$, there exists a set $C \in \calC$ satisfying $r(j) \in C$ and $C \setminus \{r(j)\} \subseteq N(j)$.
Consider the encoding function over $\Sigma = \{0,1\}$ that for every set $C \in \calC$ includes the xor of the bits associated with $C$. This encoding function forms an index code for $H$, because the $j$th receiver is able to discover its required bit given its side information and the xor associated with a set $C \in \calC$ which satisfies $r(j) \in C$ and $C \setminus \{r(j)\} \subseteq N(j)$. It thus follows that for such a collection $\calC$, it holds that $\beta(H) \leq \beta_1(H) \leq |\calC|$.
We refer to such an index code for $H$, in this generalized setting, as a clique cover index code.

We next mention a lower bound from~\cite{BlasiakKL13} on the index coding rate in the generalized setting.
An expanding sequence of size $k$ for $H$ is a sequence $j_1,\ldots, j_k \in [m]$ such that for each $\ell \in [k]$, it holds that $r(j_\ell) \notin \cup_{t<\ell}{N(j_t)}$, that is, the symbol required by receiver $j_\ell$ is unknown to receivers $j_1, \ldots, j_{\ell-1}$.
Let $\MES(H)$ denote the maximum size of an expanding sequence for $H$.
It was proved in~\cite[Lemma~III.1]{BlasiakKL13} that for every instance $H$ of the generalized index coding problem, it holds that $\beta(H) \geq \MES(H)$.

We prove the following algorithmic result for the generalized index coding problem.
Note that its guarantee is useful only for instances whose number of receivers $m$ satisfies $m = o(n \cdot \log n)$.

\begin{theorem}\label{thm:genIC}
There exists a polynomial-time algorithm that given an instance $H$ of the generalized index coding problem with $n$ symbols and $m$ receivers ($m \geq n$) finds a clique cover index code for $H$ of length $O(\frac{m}{\log n}) \cdot \MES(H) \leq O(\frac{m}{\log n}) \cdot \beta (H)$.
\end{theorem}

\begin{proof}
Let $H$ be an instance of the generalized index coding problem with $n$ symbols and $m$ receivers ($m \geq n$).
As above, for each $j \in [m]$, let $r(j) \in [n]$ denote the index of the symbol required by the $j$th receiver, and let $N(j) \subseteq [n] \setminus \{r(j)\}$ denote the set of indices of the symbols that the $j$th receiver has as side information.

Consider the algorithm that given such an instance $H$ maintains a set $A_i \subseteq [m]$ for each $i \in [n]$, initiated as the set of indices of the receivers that are interested in the $i$th symbol, namely, as $\{ j \in [m] \mid r(j)=i\}$. The $i$th symbol is called active whenever $A_i \neq \emptyset$.
The algorithm chooses for every active symbol $i$ some receiver $j_i \in A_i$ and constructs an instance $G=(V,E)$ of the (standard) index coding problem as follows. The vertex set $V$ of $G$ consists of the indices of the currently active symbols of $H$, and each vertex $i$ has a directed edge to the vertices of $N(j_i) \cap V$, i.e., the indices of the active symbols known to receiver $j_i$. Note that an index code for $G$ satisfies all the receivers $j_i$ with $i \in V$. Our algorithm calls the algorithm given by Theorem~\ref{thm:IntroICdir} on $G$ to obtain a clique cover index code for $G$ of length
\begin{eqnarray}\label{eq:GenIC_Vi}
O \Big (\frac{|V|}{\log |V|} \Big ) \cdot \MAIS(G) \leq O \Big (\frac{|V|}{\log |V|} \Big ) \cdot \MES(H),
\end{eqnarray}
where the inequality holds because the receivers associated with an acyclic induced sub-digraph of $G$, ordered appropriately, form an expanding sequence of $H$.

The algorithm proceeds by removing the receiver $j_i$ from $A_i$ for each $i \in V$ and repeating the process as long as there are at least, say, $\sqrt{n}$ active symbols. Once there are fewer than this number of active symbols, the algorithm constructs the trivial index code that transmits all the active symbols, satisfying all the remaining receivers of $\cup_{i \in [n]}{A_i}$. Note that the length of the latter is smaller than $\sqrt{n}$. Observe that the index code defined as the concatenation of all the index codes constructed by our algorithm is a valid clique cover index code for $H$. Observe further that the algorithm runs in polynomial time, because in each iteration the number of receivers decreases, and because the algorithm from Theorem~\ref{thm:IntroICdir} runs in polynomial time.

It remains to analyze the length of the index code produced by the algorithm.
For an instance $H$ with $n$ symbols and $m$ receivers, let $G_1=(V_1,E_1), \ldots, G_t=(V_t,E_t)$ denote the instances of the (standard) index coding problem constructed throughout the run of the algorithm. By the definition of the algorithm we have $|V_i| \geq \sqrt{n}$ for all $i \in [t]$. Using~\eqref{eq:GenIC_Vi}, this implies that for each $i \in [t]$, the length of the index code produced for $G_i$ is bounded by
\[O \Big (\frac{|V_i|}{\log |V_i|} \Big ) \cdot \MES(H) \leq O \Big (\frac{|V_i|}{\log n} \Big ) \cdot \MES(H).\]
Since the vertices of the digraphs $G_1, \ldots, G_t$ represent distinct receivers of $H$, their vertex sets satisfy $\sum_{i=1}^{t}{|V_i|} \leq m$.
Recalling that the last component of the constructed index code has length smaller than $\sqrt{n}$, we obtain that the total length of the index code produced by the algorithm does not exceed
\[ O \Big (\sum_{i=1}^{t}{\frac{|V_i|}{\log n}} \Big  ) \cdot \MES(H) + \sqrt{n} \leq O \Big (\frac{m}{\log n} \Big ) \cdot \MES(H),\]
where the inequality holds using $m \geq n$. This completes the proof.
\end{proof}

\section*{Acknowledgments}
We thank the anonymous referees for their constructive comments.

\bibliographystyle{abbrv}
\bibliography{approxIC}

\begin{thebibliography}{10}

\bibitem{AlonLSWH08}
N.~Alon, E.~Lubetzky, U.~Stav, A.~Weinstein, and A.~Hassidim.
\newblock Broadcasting with side information.
\newblock In {\em Proc. of the 49th Annual {IEEE} Symposium on Foundations of
  Computer Science ({FOCS}'08)}, pages 823--832, 2008.

\bibitem{AlonS16}
N.~Alon and J.~H. Spencer.
\newblock {\em The Probabilistic Method}.
\newblock Wiley Publishing, 4th edition, 2016.

\bibitem{ArbabjolfaeiK16}
F.~Arbabjolfaei and Y.~Kim.
\newblock Approximate capacity of index coding for some classes of graphs.
\newblock In {\em Proc. of the {IEEE} International Symposium on Information
  Theory ({ISIT}'16)}, pages 2154--2158, 2016.

\bibitem{ArbabjolfaeiK18}
F.~Arbabjolfaei and Y.~Kim.
\newblock Fundamentals of index coding.
\newblock {\em Found. Trends Commun. Inf. Theory}, 14(3--4):163--346, 2018.

\bibitem{AwasthiCKS15}
P.~Awasthi, M.~Charikar, R.~Krishnaswamy, and A.~K. Sinop.
\newblock The hardness of approximation of {E}uclidean $k$-means.
\newblock In {\em Proc. of the 31st International Symposium on Computational
  Geometry (SoCG'15)}, pages 754--767, 2015.

\bibitem{BBJK06}
Z.~Bar{-}Yossef, Y.~Birk, T.~S. Jayram, and T.~Kol.
\newblock Index coding with side information.
\newblock {\em {IEEE} Trans. Inform. Theory}, 57(3):1479--1494, 2011.
\newblock Preliminary version in FOCS'06.

\bibitem{BelmonteHHRS14}
R.~Belmonte, P.~Heggernes, P.~van~'t Hof, A.~Rafiey, and R.~Saei.
\newblock Graph classes and {R}amsey numbers.
\newblock {\em Discret. Appl. Math.}, 173:16--27, 2014.

\bibitem{BirkKol98}
Y.~Birk and T.~Kol.
\newblock Coding on demand by an informed source ({ISCOD}) for efficient
  broadcast of different supplemental data to caching clients.
\newblock {\em IEEE Trans. Inform. Theory}, 52(6):2825--2830, 2006.
\newblock Preliminary version in INFOCOM'98.

\bibitem{BlasiakKL13}
A.~Blasiak, R.~Kleinberg, and E.~Lubetzky.
\newblock Broadcasting with side information: Bounding and approximating the
  broadcast rate.
\newblock {\em {IEEE} Trans. Inform. Theory}, 59(9):5811--5823, 2013.

\bibitem{BoppanaH92}
R.~B. Boppana and M.~M. Halld{\'{o}}rsson.
\newblock Approximating maximum independent sets by excluding subgraphs.
\newblock {\em BIT Numer. Math.}, 32(2):180--196, 1992.
\newblock Preliminary version in SWAT'90.

\bibitem{ChawinH23}
D.~Chawin and I.~Haviv.
\newblock Improved {NP}-hardness of approximation for orthogonality dimension
  and minrank.
\newblock {\em {SIAM} J. Discret. Math.}, 37(4):2670--2688, 2023.
\newblock Preliminary version in STACS'23.

\bibitem{ChawinH22}
D.~Chawin and I.~Haviv.
\newblock Hardness of linear index coding on perturbed instances.
\newblock {\em {IEEE} Trans. Inform. Theory}, 70(2):1388--1396, 2024.
\newblock Preliminary version in Allerton'22.

\bibitem{ClawFree05}
M.~Chudnovsky and P.~D. Seymour.
\newblock The structure of claw-free graphs.
\newblock In B.~S. Webb, editor, {\em Surveys in Combinatorics 2005}, volume
  327 of {\em London Mathematical Society Lecture Note Series}, pages 153--172.
  Cambridge University Press, 2005.

\bibitem{DaneshpajouhMM21}
H.~R. Daneshpajouh, F.~Meunier, and G.~Mizrahi.
\newblock Colorings of complements of line graphs.
\newblock {\em J. Graph Theory}, 98(2):216--233, 2021.

\bibitem{Erdos67}
P.~Erd{\H{o}}s.
\newblock Some remarks on chromatic graphs.
\newblock {\em Colloq. Math.}, 16:253--256, 1967.

\bibitem{ErdosMoser64}
P.~Erd{\H{o}}s and L.~Moser.
\newblock On the representation of directed graphs as unions of orderings.
\newblock {\em Magyar Tud. Akad. Mat. Kutat\'o Int. K\"ozl.}, 9:125--132, 1964.

\bibitem{ErdosS35}
P.~Erd{\H{o}}s and G.~Szekeres.
\newblock A combinatorial problem in geometry.
\newblock {\em Compositio Mathematica}, 2:463--470, 1935.

\bibitem{Feige04}
U.~Feige.
\newblock Approximating maximum clique by removing subgraphs.
\newblock {\em {SIAM} J. Discret. Math.}, 18(2):219--225, 2004.

\bibitem{Haemers78}
W.~H. Haemers.
\newblock An upper bound for the {S}hannon capacity of a graph.
\newblock In L.~Lov\'asz and V.~T. S\'os, editors, {\em Algebraic Methods in
  Graph Theory}, volume 25/I of {\em Colloquia Mathematica Societatis J\'anos
  Bolyai}, pages 267--272. Bolyai Society and North-Holland, 1978.

\bibitem{Halldorsson93}
M.~M. Halld{\'{o}}rsson.
\newblock A still better performance guarantee for approximate graph coloring.
\newblock {\em Inf. Process. Lett.}, 45(1):19--23, 1993.

\bibitem{KornerS92}
J.~K{\"{o}}rner and G.~Simonyi.
\newblock A {S}perner-type theorem and qualitative independence.
\newblock {\em J. Comb. Theory, Ser. {A}}, 59(1):90--103, 1992.

\bibitem{LangbergS08}
M.~Langberg and A.~Sprintson.
\newblock On the hardness of approximating the network coding capacity.
\newblock {\em IEEE Trans. Inform. Theory}, 57(2):1008--1014, 2011.
\newblock Preliminary version in ISIT'08.

\bibitem{LovaszKneser}
L.~Lov{\'{a}}sz.
\newblock Kneser's conjecture, chromatic number, and homotopy.
\newblock {\em J. Comb. Theory, Ser. {A}}, 25(3):319--324, 1978.

\bibitem{Shannon56}
C.~E. Shannon.
\newblock The zero error capacity of a noisy channel.
\newblock {\em Institute of Radio Engineers, Trans. Inform. Theory},
  IT-2:8--19, 1956.

\bibitem{BookFekete}
J.~H. van Lint and R.~M. Wilson.
\newblock {\em A course in combinatorics}.
\newblock Cambridge University Press, second edition, 2001.

\bibitem{Zuckerman07}
D.~Zuckerman.
\newblock Linear degree extractors and the inapproximability of max clique and
  chromatic number.
\newblock {\em Theory of Computing}, 3(6):103--128, 2007.
\newblock Preliminary version in STOC'06.

\end{thebibliography}

\appendix

\section{Proof of Theorem~\ref{thm:BH92}}\label{app:undirected}

We show here that the algorithmic result of Boppana and Halld{\'{o}}rsson~\cite{BoppanaH92}, stated as Theorem~\ref{thm:BH92}, can be derived from our Theorem~\ref{thm:CC_general_Undir}. The proof requires the following lemma, which is given implicitly in~\cite{BoppanaH92} and relies on an argument of~\cite{ErdosS35}.

\begin{lemma}\label{lemma:RamseyU}
Let $Q: \N \times \N \rightarrow \N$ be the function defined by $Q(s,t) = \binom{s+t-2}{s-1}$.
There exists a polynomial-time algorithm that given a graph $G$ on $n$ vertices finds a clique $C$ of $G$ and an independent set $I$ of $G$, such that for all integers $s,t \geq 1$, if $n \geq Q(s,t)$ then either $|C| \geq s$ or $|I| \geq t$.
\end{lemma}

\begin{proof}
We define a recursive algorithm, called $\Ramsey$, that given a graph $G=(V,E)$ returns a pair $(C,I)$ of subsets of $V$. The algorithm is defined as follows.

\begin{enumerate}
  \item\label{itm:11} If $V = \emptyset$, then return the pair $(\emptyset,\emptyset)$.
  \item\label{itm:22} Choose an arbitrary vertex $u \in V$.
  \item\label{itm:33} $(C_1, I_1) \leftarrow \Ramsey(G[V_1])$, where $V_1 = \{ v \in V \mid \{u,v\} \in E\}$.
  \item\label{itm:44} $(C_2, I_2) \leftarrow \Ramsey(G[V_2])$, where $V_2 = V \setminus (V_1 \cup \{u\})$.
  \item\label{itm:55} Let $C$ be a largest set among $C_1 \cup \{u\}$ and $C_2$, let $I$ be a largest set among $I_1$ and $I_2 \cup \{u\}$, and return the pair $(C,I)$.
\end{enumerate}

Observe that for any choice of a vertex $u$ in Step~\ref{itm:22} of the $\Ramsey$ algorithm, the sets $V_1$ and $V_2$ defined in Steps~\ref{itm:33} and~\ref{itm:44} do not include $u$ and are disjoint. This implies that throughout the run of the algorithm on a graph $G$, every recursive call chooses in Step~\ref{itm:22} a distinct vertex of $G$.
It follows that the algorithm can be implemented in polynomial time.

We first prove that on every input graph $G=(V,E)$, the $\Ramsey$ algorithm returns a pair $(C,I)$ such that $C$ is a clique of $G$ and $I$ is an independent set of $G$.
This is shown by induction on the number of vertices of $G$. If $G$ has no vertices, this trivially holds by Step~\ref{itm:11} of the algorithm.
Otherwise, the algorithm chooses in Step~\ref{itm:22} an arbitrary vertex $u \in V$, and for each $i \in [2]$, it calls the $\Ramsey$ algorithm recursively on the subgraph $G[V_i]$ to obtain a pair $(C_i,I_i)$, where $V_1,V_2$ are defined in Steps~\ref{itm:33} and~\ref{itm:44} of the algorithm. By the inductive assumption, for each $i \in [2]$, $C_i$ is a clique of $G[V_i]$, and $I_i$ is an independent set of $G[V_i]$.

By the definition of $V_1$, the vertex $u$ is adjacent to all vertices of $V_1$, hence $C_1 \cup \{u\}$ is a clique of $G$. It thus follows that the set $C$, defined in Step~\ref{itm:55} as a largest set among $C_1 \cup \{u\}$ and $C_2$, is a clique of $G$, as required.
By the definition of $V_2$, the vertex $u$ has no edges to the vertices of $V_2$, hence $I_2 \cup \{u\}$ is an independent set of $G$.
It thus follows that the set $I$, defined in Step~\ref{itm:55} as a largest set among $I_1$ and $I_2  \cup \{u\}$, is an independent set of $G$, as required.

We finally analyze the sizes of the sets returned by the $\Ramsey$ algorithm.
To do so, for every pair of integers $s,t \geq 1$, let $P(s,t)$ denote the smallest integer $p$ such that on every graph on at least $p$ vertices, the $\Ramsey$ algorithm is guaranteed to return a clique of size at least $s$ or an independent set of size at least $t$.
We claim that for all $s,t \geq 2$, it holds that
\begin{eqnarray}\label{eq:Ramsey}
P(s,t) \leq P(s-1,t)+ P(s,t-1)-1.
\end{eqnarray}
To see this, consider a graph $G$ on $p \geq P(s-1,t)+ P(s,t-1)-1$ vertices.
While running on $G$, the $\Ramsey$ algorithm calls itself recursively on two subgraphs of $G$, induced by sets $V_1$ and $V_2$, which satisfy $|V_1|+|V_2|=p-1$. Our assumption on $p$ implies that either $|V_1| \geq P(s-1,t)$ or $|V_2| \geq P(s,t-1)$.
By the definition of $P(s,t)$, the two pairs $(C_1, I_1)$ and $(C_2, I_2)$ returned by the recursive calls satisfy $|C_1| \geq s-1$ or $|I_1| \geq t$ in the first case, and $|C_2| \geq s$ or $|I_2| \geq t-1$ in the second case.
By the definition of the algorithm, it follows that in each of these cases, the $\Ramsey$ algorithm returns on $G$ a clique of size at least $s$ or an independent set of size at least $t$, as required.

It remains to verify that $P(s,t) \leq Q(s,t)$ for all $s,t \geq 1$.
This can be checked by induction on $s+t$.
First, if either $s=1$ or $t=1$, then $P(s,t)= Q(s,t) =1$.
Next, for the induction step, combine~\eqref{eq:Ramsey} with the inductive assumption on $s+t-1$ to obtain that
\begin{eqnarray*}
P(s,t) &\leq& P(s-1,t)+ P(s,t-1)-1 \\
&\leq& Q(s-1,t)+ Q(s,t-1) \\
&=& \binom{s+t-3}{s-2} + \binom{s+t-3}{s-1} \\
&=& \binom{s+t-2}{s-1} = Q(s,t),
\end{eqnarray*}
where the second equality holds by Pascal's identity.
This completes the proof.
\end{proof}

Now, we combine Theorem~\ref{thm:CC_general_Undir} with Lemma~\ref{lemma:RamseyU} to prove Theorem~\ref{thm:BH92}.

\begin{proof}[ of Theorem~\ref{thm:BH92}]
Let $Q: \N \times \N \rightarrow \N$ be the function given in Lemma~\ref{lemma:RamseyU}, defined by $Q(s,t) = \binom{s+t-2}{s-1}$.
It can be shown (see, e.g.,~\cite[Lemma~1]{Erdos67}) that the function $f_Q$ associated with $Q$ in Theorem~\ref{thm:CC_general_Undir} satisfies
\begin{eqnarray}\label{eq:f_Q_U}
f_Q(n) \geq \Omega(\log^2 n).
\end{eqnarray}
Equipped with the algorithm given by Lemma~\ref{lemma:RamseyU}, we apply Theorem~\ref{thm:CC_general_Undir} to the family of all graphs.
It follows that there exists a polynomial-time algorithm that given a graph $G$ on $n$ vertices finds a clique cover of $G$ whose size does not exceed
\[ \Big (\sum_{i=1}^{n}{\frac{1}{f_Q(i)}} \Big ) \cdot \alpha(G) \leq  O \Big (\frac{n}{\log^2 n} \Big ) \cdot \alpha(G),\]
where the above inequality can be verified using~\eqref{eq:f_Q_U}. This completes the proof.
\end{proof}

\end{document}